\documentclass[12pt,reqno]{amsart}
\usepackage{amsthm,amsfonts,amssymb,euscript}

\newcommand{\bea}{\begin{eqnarray}}
\newcommand{\eea}{\end{eqnarray}}
\def\beaa{\begin{eqnarray*}}
\def\eeaa{\end{eqnarray*}}
\def\ba{\begin{array}}
\def\ea{\end{array}}
\def\be#1{\begin{equation} \label{#1}}
\def \eeq{\end{equation}}

\def\a{{\alpha}}

\def\b{{\beta}}
\def\be{{\beta}}
\def\ga{\gamma}
\def\Ga{\Gamma}
\def\de{\delta}

\def\la{\lambda}

\def\ze{\zeta}

\def\p{{\partial}}

\def\LL{{\mathcal L}}

\def\I{{\bf I}}

\def\M{{\bf M}}
\def\N{{\bf N}}

\def\O{{\bf O}}

\def\S{{\bf S}}
\def\K{{\bf K}}

\def\g{{\bf g}}

\def\f12{{\frac 1 2}}
\def\ub{\underline{u}}

\def\Lb{{\,\underline{L}}}

\def\chib{{\underline \chi}}

\def\f{\widetilde{f}}

\newtheorem{theorem}{Theorem}[section]
\newtheorem{lemma}[theorem]{Lemma}
\newtheorem{claim}[theorem]{Claim}
\newtheorem{proposition}[theorem]{Proposition}

\newtheorem{definition}[theorem]{Definition}
\newtheorem{remark}[theorem]{Remark}

\numberwithin{equation}{section}

\begin{document}\title[On Hawking's Local Rigidity Theorems for Charged Black Holes]{On Hawking's Local Rigidity Theorems for Charged Black Holes}
\author{Pin Yu}
\address{Princeton University}
\email{pinyu@math.princeton.edu}

\begin{abstract}
We show the existence of a Hawking vector field in a full
neighborhood of a local, regular, bifurcate, non-expanding horizon
embedded in a smooth Einstein-Maxwell space-time without assuming
the underlying space-time is analytic. It extends one result of
Friedrich, R\'{a}cz and Wald, see \cite{FRW}, which was limited to
the interior of the black hole region. Moreover, we also show, in
the presence of an additional Killing vector field $T$ which
tangent to the horizon and not vanishing on the bifurcate sphere,
then space-time must be locally axially symmetric without the
analyticity assumption. This axial symmetry plays a fundamental
role in the classification theory of stationary black holes.

\end{abstract}
\maketitle

\section{Introduction}\label{introduction}
Let $(\M, g, F)$ be a smooth and time oriented Einstein-Maxwell
space-time of dimension $3+1$ with electromagnetic field $F$. Let
$\S$ be an smoothly embedded space-like $2$-sphere in $\M$ and
$\N^+$, $\N^-$ be the corresponding null boundaries of the causal
future and the causal past of $\S$. We also assume that both
$\N^+$ and $\N^-$ are regular, achronal, null hypersurfaces in a
neighborhood $\O$ of $\S$. The triplet $(\S, \N^+, \N^-)$ is
called a local, regular bifurcate horizon in $\O$. The main result
of the paper asserts if $(\S, \N^+, \N^-)$ is non-expanding (see
Definition \ref{nonexpansion}), then it must be a Killing
bifurcate horizon. More precisely, we have the following theorem:

\begin{theorem}\label{first}Given a local, regular, bifurcate,
non-expanding horizon $(\S, \N^+, \N^-)$ in a smooth and time
oriented Einstein-Maxwell space-time $(\O, g, F)$, there exists an
neighborhood $\O' \subset \O$ of $S$ and a non-trivial Killing
vector field $K$ in $\O'$, which is tangent to the null generators
of $\N^+$ and $\N^-$. Moreover, the Lie derivative $\LL_K F =0$.
\end{theorem}

The vector field $K$ is called the Hawking vector field in the
literature. Its existence is already known under the assumption
that the space-time is real analytic. In the work of \cite{FRW},
the authors showed, by solving wave equations, the existence of
Hawking vector field $K$ without the analyticity assumption, but
$K$ could only be constructed inside the domain of dependence of
$\N^+ \cup \N^-$ due to the fact that the corresponding wave
equations are ill-posed outside this region. So the new ingredient
of our theorem is to extend the Hawking vector field $K$ to a full
neighborhood of the bifurcate sphere $S$, without making any
additional regularity assumptions on the underlying space-time
$(\M, g)$.  We use the idea of S. Alexakis, A. Ionescu and S.
Klainerman, who proved a similar theorem for Einstein vacuum
space-time, see \cite{AIK} for details.

We also prove the following theorem:

\begin{theorem}\label{second}
Given a local, regular, bifurcate horizon $(\S, \N^+, \N^-)$ in a
smooth and time oriented Einstein-Maxwell space-time $(\O, g, F)$.
If there is a Killing vector field $T$ tangent to $\N^+ \cup \N^-$
and non-vanishing on $\S$. Then there is a neighborhood $\O'
\subset \O$ of $\S$, such that we can find a rotational Killing
vector $Z$ in $\O'$, i.e. $Z$ has closed orbits. Moreover, $[Z, T]
= 0$. If in addition $\LL_T F =0$, then $\LL_Z F =0$.
\end{theorem}

Although, we don't make the non-expanding assumption on the
horizon, it's a well known fact that the non-expansion is a
consequence of the fact that  the Killing vector field $T$ is
tangent to $\N^+ \cup \N^-$. So the first theorem will produce a
Hawking vector field $K$ in a full neighborhood of $S$. The
rotational vector field can be written as a linear combination of
$T$ and $K$, i.e. we show that the existence a constant $\la$ such
that
$$Z = T + \la K$$
is a rotation with period $t_0$. So the part $\LL_Z F =0$ in the
theorem follows immediately. In the proof, we will focus on other
parts of the theorem. The period $t_0$ is determined on the
bifurcate sphere $\S$, while to determine $\la$, we need the
information on $S$ and the information of one particular null
geodesic on $\N^+ \cup \N^-$, see the proof for more details.

Once more, under the restrictive additional assumption of
analyticity of the space-time $(\M, g)$, this second theorem is
also known for Einstein vacuum space-times. It's usually called
Hawking's rigidity theorem, see \cite{HE}, which asserts that
under some global causality, asymptotic flatness and connectivity
assumptions, a stationary, non-degenerate analytic space-time must
be axially symmetric. In the smooth category, one can find a proof
in \cite{AIK} based on the idea that, under a suitable conformal
rescaling of null generators on the bifurcate sphere, the level
sets of the affine parameters of the null generators on the
horizon should represent the integrable surface ruled out by the
closed rotational orbits. We will give a more geometric
construction.

These two theorems play an important role in the classification
theory of stationary black holes, since they reduce the
classifications to the cases which are covered by the well-known
uniqueness theorems for electrovac black holes in general
relativity, see \cite{I}, \cite{HE}.\\

We now describe the main ideas of the proofs. The first step is to
construct the Hawking vector field $K$. Since $K$ is a Killing
vector field, it must satisfy the following covariant linear wave
equations:
\begin{equation}\label{equationK}
 \square_g K_\a = -R_\a{}^\b K_\b
\end{equation} 
where $R_{\a\b}$ is the Ricci curvature tensor for the Lorentzian
metric $g$. We hope to reconstruct $K$ by solving this wave
equation. This is precisely the strategy used in \cite{FRW}. The
equation can be solved in the domain of dependence if initial data
is prescribed on the characteristic hypersurfaces, see
\cite{Rendall} for a proof. The choice of initial data can be
rediscovered by the following heuristic argument: because $K$ is
Killing, its restriction on a geodesic should be a Jacobi field,
so it's reasonable to guess the initial data on $\N^+$ should be
the non-trivial parallel Jacobi field $\ub L$ where $L$ is one
null geodesic generator on $\N^+$ and $\ub$ is the corresponding
affine parameter, i.e. $L(\ub)=1$; another way to guess the
initial data is to check the explicit formula for the exact
Kerr-Newman solutions. While the Cauchy problem for
(\ref{equationK}) is ill-posed on complement of the domain of
dependence, solving (\ref{equationK}) can not construct the
Hawking vector field in bad region. We have to rely on the new
techniques used in \cite{AIK}. A careful calculation shows $K$
also solves an ordinary differential equation which is well-posed
in the ill-posed region for (\ref{equationK}). So one can extend
$K$ into the bad region by solving this ordinary equation. That's
how we construct $K$ in a full neighborhood of $\S$. Notice that
although $K$ is constructed, it's not automatically a Killing
vector field. One turns to prove the one parameter group $\phi_t$
generated by $K$ acts isometrically. We need to show that, for
each small $t$, the pull-back metric $\phi_t^* g$ must coincide
with $g$, in view of the fact that they are both solutions of
Einstein-Maxwell equations and coincide on $\N^+ \cup \N^-$. Now
the uniqueness for metric type problems come into play. The
results of Ionescu-Klainerman \cite{IK}, \cite{IK2}, Alexakis
\cite{Al} and Alexakis-Ionescu-Klainerman
\cite{AIK} provide hints to the answer.\\

The paper is organized as follows. In section 2, we construct a
canonical null frame associated to the bifurcate horizon $(\S,
\N^+, \N^-)$ and derive a set of partial differential equations
for various geometric quantities, as consequences of non-expasion
condition and the Einstein-Maxwell equations; in section 3, we
give a self-contained proof of Theorem \ref{first} in the domain
of dependence of $\N^+ \cup \N^-$, which is the Proposition B.1 in
\cite{FRW}; in section 4, based on the Carleman estimates proved
in \cite{IK} and \cite{IK2}, we extend the Hawking vector field to
a full neighborhood of $\S$ which completes the proof of Theorem
\ref{first}.; the last section is devoted to a geometric proof of
Theorem \ref{second}.

\medskip

{\bf Acknowledgements}:\quad The author would like to thank
Professor Sergiu Klainerman for suggesting the problem; and Willie
Wai-Yeung Wong for valuable discussions.

\section{Preliminaries}\label{prelim}
In this paper, the indices $\a,\b,\ga,\de,\rho$ are from $1$ to
$4$, $a,b,c$ are from $1$ to $2$; the curvature convention is
$R_{\a\b\ga\de} = g(D_\a D_\b e_\ga-D_\b D_\a e_\ga, e_\de)$,
where $D_\a D_\b X = D_\a( D_\b X)-D_{D_\a e_\b}X$; repeat indices
are always understood as Einstein summation convention; since
during the proof of our main theorems, we will keep shrinking the
open neighborhood $\O$ of $\S$ mentioned in the introduction, we
keep denoting such neighborhoods by $\O$ for simplicity.

One can choose a smooth future-directed null pair $(L,\Lb)$ along
$\S$ with normalization
\begin{equation*}
 g(L,L)=g(\Lb,\Lb)=0,\quad g(L,\Lb)=-1
\end{equation*}
such that $L$ is tangent to $\N^+$ and $\Lb$ is tangent to $\N^-$.
In a small neighborhood of $\S$, we extend $L$ along the null
geodesic generators of $\N^+$ via parallel transport; we also
extend $\Lb$ along the null geodesic generators of $\N^-$ via
parallel transport. So $D_L L = 0$ and $D_\Lb \Lb =0$. We now
define two optical functions $u$ and $\ub$ near $\S$. The function
$\ub$ (resp. $u$) is defined along $\N^+$ (resp. $\N^-$) by
setting initial value $\ub =0$(resp. $u=0$) on $\S$ and solving
$L(\ub)=1$ (resp. $\Lb(u)=1$). Let $\S_{\ub}$ (resp. $\S_u$) be
the level surfaces of $\ub$(resp. $u$) along $\N^+$ (resp.$\N^-$).
We define $\Lb$ (resp. $L$) on each point of the hypersurface
$\N^+$ (resp. $\N^-$) to be unique, future directed null vector
orthogonal to the surface $\S_{\ub}$ (resp. $\S_{u}$) passing
though that point and such that $g(L,\Lb)=-1$. The null
hypersurface $N^-_{\ub}$ (resp. $N^+_u$) is defined to be the
congruence of null geodesics initiating on $\S_{\ub} \subset \N^+$
(resp. $\S_{u} \subset \N^-$)in the direction of $\Lb$ (resp.
$L$). We require the null hypersurfaces $N^-_{\ub}$ (resp.
$N^+_{u}$) are the level sets of the function $\ub$ (resp. $u$),
by this condition, $u$ and $\ub$ are extended into a neighborhood
of $\S$ from the null hypersurface $\N^+ \cup \N^-$. The we can
extend both $L$ and $\Lb$ into a neighborhood of $\S$ as gradients
of the optical functions
\begin{equation*}
  L=-\g^{\mu\nu}\p_\mu u\p_\nu,\quad  \Lb=-\g^{\mu\nu}\p_\mu\ub \p_\nu.
\end{equation*}
Since $u$ and $\ub$ are null optical functions, we know
\begin{equation*}
g(L,L)=g(\Lb,\Lb)=0
\end{equation*}
while $g(L,\Lb)=-1$ only holds on the null surface $\N^+ \cup
\N^-$. Moreover, we have
\begin{equation*}
L(\ub)=1 \quad \mbox{on} \quad \N^+, \qquad \Lb(u)=1 \quad
\mbox{on} \quad \N^-.
\end{equation*}
 We define $S_{u\ub} = N^+_u \cap N^-_{\ub}$.
Using the null pair $(L,\Lb)$ one can choose a null frame $\{e_1.
e_2, e_3 = \Lb, e_4 =L\}$ such that
$$g(e_a,e_b)=\delta_{ab}, \qquad g(e_a,e_3)=g(e_a,e_4)=0,\quad
a,b=1,2.$$ At each point $p\in S_{u\ub} \subset \O$, $e_1, e_2$
form an orthonormal frame along the 2-surface $S_{u\ub}$. We will
modify the frame by Fermi transport later. Recall the null second
fundamental forms $\chi$, $\chib$ and torsion $\ze$ are defined on
$\N^+ \cup \N^-$ via the given null pair $(L,\Lb)$:
\begin{equation*}
\chi_{ab} = g(D_{e_a}L,e_b), \quad \chib_{ab} = g(D_{e_a}\Lb,e_b),
\quad \ze_a = g(D_{e_a}L,\Lb).
\end{equation*}
The traces of $\chi$ is defined by $tr\chi = \chi^a{}_a$,
similarly for $tr\chib$
\begin{definition}\label{nonexpansion}
We say that $\N^+$ is non-expanding if $tr\chi=0$ on $\N^+$;
similarly  $\N^-$ is non-expanding if  $tr\chib=0$ on $\N^-$. The
bifurcate horizon $(\S, \N^+, \N^-)$ is called non-expanding if
both $\N^+, \N^-$ are non-expanding.
\end{definition}
The non-expansion condition has a very strong restriction on the
geometry of the Einstein-Maxwell space-time. We recall the
Einstein-Maxwell equations:
\begin{equation*}
\left\{ \begin{array}{rl}
    R_{\a\b} - \frac{1}{2} R g_{\a\b} &= T_{\a\b}\\
        D_{[\a}F_{\b\ga]} &=0\\
    D^\a F_{\a\b}&=0
        \end{array}\right.
\end{equation*}
where $T_{\a\b} = F_\a{}^\mu
F_{\b\mu}-\frac{1}{4}g_{\a\b}F^{\mu\nu}F_{\mu\nu}$ is the
energy-momentum tensor for the corresponding electromagnetic
field. Since the dimension of the underlying manifold is $4$, the
field theory is conformal, i.e. $tr T =0$. So by tracing the first
equation in the system, we know the scalar curvature $R=0$. We can
rewrite the system as
\begin{equation}\label{E-M}
\left\{ \begin{array}{rl}
    R_{\a\b} &= F_\a{}^\mu F_{\b\mu}-\frac{1}{4}g_{\a\b}F^{\mu\nu}F_{\mu\nu}\\
        D_{[\a}F_{\b\ga]} &=0\\
    D^\a F_{\a\b}&=0
        \end{array}\right.
\end{equation}
We recall that the positive energy condition is valid for
Einstein-Maxwell energy-momentum tensor, i.e.
$$T(X, Y)\geq 0$$
where $(X, Y)$ are an arbitrary pair of future-directed causal
vectors. Let $\hat{\chi}$ be the traceless part of $\chi$, so on
$\N^+$, according to Raychaudhuri equation:
\begin{equation*}
L(tr\chi)=-R_{LL}-|\hat{\chi}|^2-\frac{1}{2}(tr\chi)^2
\end{equation*}
So non-expansion condition on the black hole boundary implies
\begin{equation*}
R_{LL}+|\hat{\chi}|^2 =0
\end{equation*}
One can take advantage of the positive energy condition to
conclude
\begin{equation*}
R_{LL}=0, \qquad \hat{\chi} = 0 \qquad \text{on} \quad \N^+.
\end{equation*}
So $\chi = 0$ on $\N^+$. According to untraced formulation of
Raychaudhuri equation:
\begin{equation*}
L(\chi)+\chi^2+R(-,L)L=0
\end{equation*}
we know for all $X \in T\N^+$,
$$R(X,L)L=0$$
In view of the first equation in (\ref{E-M}), $R_{LL}=0$ implies
$F_{4a}=0$, and this last vanishing quantities imply $R_{4a}=0$,
combined with $R(X,L)L=0$, we know $R_{4aba}=0$. To summarize, the
non-expansion condition implies, on the null hypersurface $\N^+$
\begin{equation}\label{nonexpanding}
\left\{ \begin{array}{rl}
    \chi &= 0\\
        R_{4a}&=0\\
    R_{4aba}&=0\\
    R_{344a} &=0\\
    F_{4a}&=0
        \end{array}\right.
\end{equation}
Similar identities hold on $\N^-$ by replacing the index $4$ by
$3$. It's precisely this set of geometric information that we use
in the proof of our main theorems. Recall also our choice of the
frame ${e_1,e_2}$ is arbitrary on $\N^+$. Since we know $\chi =0$,
we can make this choice more rigid by using Fermi transport along
$L$, i.e. we first pick up an local orthonormal basis on $\S$, the
use the Lie transport relation $\LL _L e_a=0$ to get a basis on
$\S_{\ub}$ (which needs not to be orthonormal), the vanishing of
$\chi$ on $\N^+$ guarantees $\{e_1,e_2\}$ is still an orthonormal
basis. We summarize the computation formulas in the null frame
$\{e_1. e_2, e_3 = \Lb, e_4 =L\}$ on $\N^+$:
\begin{equation}\label{Christoffel}
\left\{ \begin{array}{rl}
    D_L L =0, \qquad &D_{e_a} L = -\ze_a L, \\
        D_L \Lb = -\ze_a e_a \quad& D_{e_a} \Lb = \chib_{ab}e_b + \ze_a \Lb\\
    D_L e_a = -\ze_a L  \quad& D_{e_a}e_b = \nabla_{e_a}e_b + \chib_{ab} L
        \end{array}\right.
\end{equation}
where $\nabla_{e_a}e_b$ is the projection of $D_{e_a}e_b$ onto the
surface $\S_u$. A similar set of identities hold on $\N^-$.

\begin{lemma}\label{Rllb}
On $\N^+$, we have
\begin{equation*}
R(-,L,\Lb,-)=-D\ze - \nabla_L \chib + \ze \otimes \ze
\end{equation*}
i.e. for all $X,Y \in T\S_{\ub}$,
\begin{equation*}
R(X,L,\Lb,Y)=-(D\ze)(X,Y) - (\nabla_L \chib)(X,Y) + \ze(X)\ze(Y).
\end{equation*}
where $\nabla$ denotes the restriction of $D$ on $\S_{\ub}$;
similar result holds on $\N^-$.
\end{lemma}

\begin{proof}
For $X,Y \in T\S_{\ub}$, we have
\begin{align*}
 R(X,L,\Lb,Y) &=g(D_X D_L \Lb,Y)-g(D_L D_X \Lb, Y)-g(D_{D_X L}\Lb,Y)+g(D_{D_L X}\Lb,Y)\\
          &=g(D_X \ze^\sharp, Y)-g(D_L(\chi(X)+\ze(X)\Lb),Y)\\
          &\quad +\ze(X)g(D_L \Lb, Y)+g(D_{\nabla_L X-\ze(X)L}\Lb,Y)\\
          &=-(D\ze)(X,Y)-g(D_L(\chib(X)),Y)+g(D_{\nabla_L X}\Lb,Y)-\ze(X)g(D_L \Lb,Y)\\
          &=-(D\ze)(X,Y) - (\nabla_L \chib)(X,Y) + \ze(X)\ze(Y).
\end{align*}
\end{proof}
\section{Hawking vector field inside black hole} \label{inside}
We define the following  four regions $\I^{++}$, $\I^{--}$, $\I^{+-}$ and $\I^{-+}$:
\begin{equation}
\begin{split}
&\I^{++}=\{p\in\O | u(p)\geq 0\text{ and }\ub(p)\geq 0\},\quad \I^{--}=\{p\in\O|u(p)\leq 0\text{ and }\ub(p)\leq 0\},\\
&\I^{+-}=\{p\in\O | u(p)\geq 0\text{ and }\ub(p)\leq 0\},\quad \I^{-+}=\{p\in\O|u(p)\leq 0\text{ and }\ub(p)\geq 0\}.
\end{split}
\end{equation}
In this section, we will prove the following proposition
\begin{proposition}\label{interior}
Under the assumptions of Theorem \ref{first}, in a small
neighborhood  $\O$ of $\S$, there exists a smooth Killing vector
field $K$ in $\O \cap (\I^{++} \cup \I^{--} )$ such that
\begin{equation*}
K=\ub L-u\Lb \quad\text{ on }(\N^+ \cup \N^-) \cap \O.
\end{equation*}
Moreover, $\LL_K F =0$ and $[\Lb,K]=-\Lb$.
\end{proposition}
The region $\O \cap (\I^{++} \cup \I^{--})$ is the domain of
dependence of $\N^+ \cup \N^-$. As we mentioned in the
introduction, by using the Newman-Penrose formalism, the first
part of the proposition is shown by H. Friedrich, I. R\'{a}cz and
R. Wald, see \cite{FRW}. For the sake of completeness, we provide
a direct proof without Newman-Penrose formalism. As mentioned in
the introduction, we consider the following characteristic initial
value problem
\begin{equation} \label{eqK}
 \left\{ \begin{array}{rl}
    \square_g K_\a &= -R_\a{}^\b K_\b \\
        K&=\ub L-u\Lb \quad\text{ on }(\N^+ \cup \N^-) \cap \O
        \end{array}\right.
\end{equation}
According to \cite{Rendall}, it's well-posed in $\O \cap (\I^{++}
\cup \I^{--})$. So a smooth vector field $K$ is now constructed in
the domain of dependence of $\N^+ \cup \N^-$. To show $K$ is
indeed a Killing vector field, one has to show the deformation
tensor of $K$
\begin{equation*}
\pi_{\a\b} = \LL_K g =D_\a K_\b + D_\b K_\a
\end{equation*}
is zero in $\O \cap (\I^{++} \cup \I^{--})$.

Since $K$ solves (\ref{eqK}), by commuting derivatives, we know
the deformation tensor $\pi_{\a\b}$ solves the following covariant
wave equation:
\begin{equation*}
\square_g \pi_{\a\b}=-2R^\rho{}_{\a\b}{}^\de
\pi_{\rho\de}+R_{\a\rho}\pi^{\rho}{}_\b+R_{\b\rho}\pi^{\rho}{}_\a
-2\LL_K R_{\a\b}
\end{equation*}
The geometric part of Einstein-Maxwell equations (\ref{E-M})
provides
\begin{align*}\label{LKR}
 \LL_K R_{\a\b} &= \LL_K T_{\a\b}\\
        &= F_\a{}^\rho \LL_K F_{\b\rho} + F_\b{}^\rho \LL_K F_{\a\rho} -\pi_{\rho\de}F_\a{}^\rho F_\b{}^\de\\
        &\quad -\frac{1}{4}\pi_{\a\b}F_{\mu\nu} F^{\mu\nu}-\frac{1}{2}g_{\a\b} F^{\mu\nu} \LL_K F_{\mu\nu}+\frac{1}{2}g_{\a\b}\pi_{\rho\de}F^\de{}_\ga F^{\rho\ga}
\end{align*}
This formula requires one to consider the partial differential
equations satisfied by $\LL_K F_{\a\b}$, which follows directly
from the electromagnetic part of the Einstein-Maxwell equations
(\ref{E-M}):
\begin{equation*}
\left\{ \begin{array}{rl}
        D_{[\a} \LL_K F_{\b\ga]} &=0\\
    D^\a \LL_K F_{\a\b} &=\pi_{\a\ga}D^\ga F^\a{}_\b + \frac{1}{2}(D_\a \pi_{\b\ga}+D_\b \pi_{\a\ga}-D_\ga \pi_{\a\b})
        \end{array}\right.
\end{equation*}
Put all the equations together, we know $\pi_{\a\b}$ and $\LL_K
F_{\a\b}$ solve the characteristic initial value problem for the
following closed symmetric hyperbolic system:
\begin{equation}\label{pi}
 \left\{ \begin{array}{rl}
        \square_g \pi_{\a\b} &=-2R^\rho{}_{\a\b}{}^\de \pi_{\rho\de}+R_{\a\rho}\pi^{\rho}{}_\b+R_{\b\rho}\pi^{\rho}{}_\a \\
                 &\quad -2(F_\a{}^\rho \LL_K F_{\b\rho} + F_\b{}^\rho \LL_K F_{\a\rho} -\pi_{\rho\de}F_\a{}^\rho F_\b{}^\de)\\
                 &\quad +\frac{1}{2}\pi_{\a\b}F_{\mu\nu} F^{\mu\nu}+g_{\a\b} F^{\mu\nu} \LL_K F_{\mu\nu}-g_{\a\b}\pi_{\rho\de}F^\de{}_\ga F^{\rho\ga}\\
    D_{[\a} \LL_K F_{\b\ga]} &=0\\
    D^\a \LL_K F_{\a\b} &=\pi_{\a\ga}D^\ga F^\a{}_\b + \frac{1}{2}(D_\a \pi_{\b\ga}+D_\b \pi_{\a\ga}-D_\ga \pi_{\a\b})
        \end{array}\right.
\end{equation}
So to show $\pi_{\a\b}=0$ and $\LL_K F=0$ in $\O$, it suffices to show
\begin{equation}\label{bdry}
 \pi_{\a\b}=0 \qquad \LL_K F=0 \qquad \text{on} \quad \N^+ \cup \N^-.
\end{equation}
We only check (\ref{bdry}) on $\N^+$; on $\N^-$, the argument is
exactly the same. In view of the expression of $K = \ub L$ on
$\N^+$ (since $u=0$ on it) and (\ref{Christoffel}), it's easy to
see
\begin{equation}\label{DK}
\left\{ \begin{array}{rl}
    D_a K_b = D_4 K_a &= D_a K_4 = D_4 K_4=0, \qquad D_4 K_3 =-1\\
        D_c D_a K_b = D_4 D_a K_b &= D_b D_4 K_a = D_4 D_4 K_a = D_a D_b K_4 = 0\\
    D_4 D_a K_4 = D_a D_4 K_4 &= D_4 D_4 K_4 = D_4 D_4 K_3 = D_a D_4 K_3 =0.
        \end{array}\right.
\end{equation}
So one knows each component of $\pi_{\a\b}$, which does not have
the bad direction $\Lb$, is zero, i.e.
\begin{equation}\label{gooddirection}
 \pi_{ab}=\pi_{4a}=\pi_{44}=0 \quad \text{on} \quad \N^+
\end{equation}
To prove the remaining components of $\pi$ vanish, we need to make
a serious use of (\ref{eqK}) to get derivatives in $\Lb$
direction. The equation (\ref{eqK}) gives
$$D_3 D_4 K_\b + D_4 D_3 K_\b =\sum_{a=1}^{2} D_a D_a K_\b + R_\b{}^\rho K_\rho$$
Combine this with curvature identity $D_3 D_4 K_\b - D_4 D_3 K_\b
= -R_{34\b}{}^\rho K_\rho$, then we have
\begin{equation}\label{eqKL}
 2 D_4 D_3 K_\b =\sum_{a=1}^{2} D_a D_a K_\b + R_{\b\rho} K^\rho + R_{34\b\rho} K^\rho
\end{equation}

\begin{claim}\label{pi34}
 \qquad We have \quad $D_3 K_4=1, \quad  D_a D_3 K_4 = D_4 D_3 K_4=0$.
\end{claim}

\begin{proof}
We set $\b = 4$ in (\ref{eqKL}), it's easy to check the left hand
side of (\ref{eqKL}) is
\begin{equation*}
2 D_4 D_3 K_4 = 2L(D_3 K_4)
\end{equation*}
while the right hand side is $0$ by (\ref{nonexpanding}). So
$L(D_3 K_4) =0$ on $\N^+$. It implies the value of $D_3 K_4$ on
$\N^+$ is determined by its value on $\S$ which is $1$. The other
identities are also easy to check, this completes the proof of the
claim.
\end{proof}

Apparently, Claim \ref{pi34} implies $\pi_{34}=0$.

\begin{claim}\label{pi3a}
 \qquad We have \quad $D_a K_3= \ub \ze_a, \quad  D_3 K_a = -\ub \ze_a$.
\end{claim}

\begin{proof}
The first identity in the claim is easy to verify by direct
computations; we now prove the second one. We first prove that
\begin{equation}\label{Lzeta}
 L(\ze_a)=0.
\end{equation}
We use (\ref{nonexpanding}):
\begin{align}
  L(\ze_a) &= L(g(D_a L,\Lb))=g(D_a L,D_L \Lb)+g(D_L D_a L,\Lb)\\
       &= g(D_L D_a L,\Lb) = R_{LaL\Lb} = 0
\end{align}
We now set $\b = b$ in (\ref{eqKL}), it implies $D_4 D_3 K_b =0$,
then by using the fact that $D_3 K_4=1$, we can show
$$L(D_3 K_a)=-\ze_a$$
Combined with (\ref{Lzeta}), it shows $ D_3 K_a = -\ub \ze_a$.
\end{proof}
Apparently, Claim \ref{pi3a} implies $\pi_{3a}=0$.

\begin{claim}\label{pi33}
 \qquad We have \quad $\pi_{33} = 2 D_3 K_3=0$.
\end{claim}
\begin{proof}

Before proving the claim, one needs more support from the
Einstein-Maxwell equations (\ref{E-M}). Since $F_{4a}=o0$, we have
\begin{equation*}
\Lb (R_{44}) = \Lb (F_{4a}^2) = 2 F_{4a} \Lb(F_{4a}) =0
\end{equation*}
which implies
\begin{equation}\label{LbR}
\Lb (R_{4aa4}) = 0
\end{equation}
Recall one of the second Bianchi identities:
\begin{equation}\label{bianchi}
D_4 R_{3aa4} + D_3 R_{a4a4} + D_a R_{43a4}=0
\end{equation}
A simple computation with the help (\ref{nonexpanding}) and
(\ref{LbR}) shows the last two terms in (\ref{bianchi}) are
zeroes. So we have
\begin{equation}\label{LR3aa4}
 L(R_{3aa4}) = D_4 R_{3aa4} =0.
\end{equation}
We compute $L(tr \chib)$ along $\N^+$:
\begin{align*}
 L(tr \chib) &= L(g(D_a \Lb,e_a))=g(D_L D_a \Lb,e_a)+g(D_a \Lb,D_L e_a)\\
         &= R_{L a \Lb a} + g(D_a D_L \Lb, e_a) + |\ze|^2\\
         &= R_{4a3a}+|\ze|^2-g(D_a \ze^\sharp, e_a)
\end{align*}
In view of (\ref{LR3aa4}) and (\ref{Lzeta}), we have
\begin{align*}
 L L(tr \chib) &= -L (g(D_a \ze^\sharp, e_a)) = -g(D_L D_a \ze^\sharp, e_a)\\
           &= - R_{4a \ze^\sharp a}-g(D_a D_L \ze^\sharp, e_a)\\
           &= -g(D_a D_L (\ze_b e_b), e_a)\\
           &= -\ze_b g(D_a (\ze_b L), e_a)
\end{align*}
This shows
\begin{equation}\label{LLtrchib}
 L L(tr \chib) = 0.
\end{equation}
Now we are ready to prove the claim. We set $\b =3$ in (\ref{eqKL}), so
\begin{align*}
 2D_4 D_3 K_3 &= D_a D_a K_3 + R_{3\rho} K^\rho + R_{343\rho} K^\rho\\
          &= D_a D_a K_3 + \ub R_{34} +  \ub R_{3434}\\
          &= D_a D_a K_3 + \ub R_{3aa4}
\end{align*}
By Lemma \ref{Rllb}, we have
\begin{align*}
 R_{3aa4} &= -(D\ze)(e_a,e_a)-(\nabla_L \chib)(e_a,e_a) + \ze_a^2\\
      &= -(D_{e_a}\ze)(e_a) -L(tr \chib) + |\ze|^2\\
      &= - div \ze + \ze(\nabla_{e_a} e_a) -L(tr \chib) + |\ze|^2
\end{align*}
We also can compute
$$D_a D_a K_3 = \ub (div \ze - \ze(\nabla_{e_a} e_a)-|\ze|^2)+ tr\chib$$
The previous computations showed
$$2 D_4 D_3 K_3 = tr \chib -\ub L(tr \chib)$$
So in view of \ref{LLtrchib}
\begin{equation}\label{LD4D3K3}
 L(D_4 D_3 K_3) = -\ub LL(tr\chib) = 0
\end{equation}
Since on $\S$, on check easily that $D_4 D_3 K_3 = 0$, so $D_4 D_3
K_3 = 0$ on $\N^+$, which once again implies $D_3 K_3 = 0$ by
solving transport equations along $L$.
\end{proof}

So we proved $\pi_{\a\b} = 0$ on $\N^+$. One still needs to show
$\LL_K F_{\a\b}=0$.
\begin{claim} \label{DDF}
We have the following identities:
\begin{equation} \label{DF}
 D_a F_{Lb} = D_L F_{Lb} = D_L F_{ab} = D_L F_{L \Lb} = 0
\end{equation}
\end{claim}

\begin{proof}
We will use (\ref{nonexpanding}) repeatedly:
\begin{align*}
 D_a F_{Lb} &= (D_a F)(L \otimes e_b)\\
        &= e_a( F_{Lb})-F(D_a L \otimes e_b)- F(L \otimes D_a e_b)\\
        &= 0.
\end{align*}
Same argument shows $D_L F_{Lb} =0$. We use Bianchi identity:
\begin{equation*}
D_L F_{ab} = -D_a F_{bL}-D_b F_{L a} =0.
\end{equation*}
We now use the last equality in Einstein-Maxwell equations
(\ref{E-M}):
\begin{equation*}
D^\a F_{\a L}=0 \Rightarrow D_a F_{aL}-D_L F_{\Lb L} =0
\end{equation*}
so $D_L F_{\Lb L} = D_a F_{aL} = 0.$
\end{proof}

\begin{claim}\label{LKFF}
On $N^+$, we have
 \begin{equation}
\LL_K F =0
 \end{equation}
\end{claim}
\begin{proof}
Recall that
\begin{equation*}
\LL_K F_{\a\b} = D_K F_{\a\b} + g^{\rho\de}D_\a K_\de F_{\rho \b}+
g^{\rho\de}D_\b K_\de F_{\a \rho}.
\end{equation*}

We show each component of $\LL_K F$ vanishes on $\N^+$:
\begin{align*}
 \LL_K F_{ab} &= D_K F_{ab} + g^{\rho\de}D_a K_\de  F_{\rho b}+ g^{\rho\de}D_b K_\de F_{a \rho}\\
          &= \ub D_L F_{ab} = 0.\\
 \LL_K F_{aL} &= D_K F_{aL} + g^{\rho\de}D_a K_\de F_{\rho L}+ g^{\rho\de}D_L K_\de F_{a \rho}\\
              &= \ub D_L F_{aL} = 0.\\
\LL_K F_{L\Lb} &= D_K F_{L\Lb} + g^{\rho\de}D_L K_\de F_{\rho \Lb}+ g^{\rho\de}D_\Lb K_\de F_{L \rho}\\
           &= D_K F_{L\Lb} -D_L K_\Lb F_{L\Lb}-D_\Lb K_L F_{L\Lb}\\
               &= \ub D_L F_{L\Lb} -\pi_{L\Lb}F_{L\Lb}=0
\end{align*}
We need some preparations to show the most difficult term $\LL_K
F_{\Lb a}$ vanishes. From the electromagnetic part of the
Einstein-Maxwell equations (\ref{E-M}), we have
\begin{equation*}
D_L F_{\Lb b} -D_\Lb F_{Lb}+D_b F_{L\Lb}=0
\end{equation*}
\begin{equation*}
-(D_L F_{\Lb b} +D_\Lb F_{Lb})+D_a F_{ab}=0
\end{equation*}
So one derives
\begin{equation}\label{LFLbb}
 2 D_L F_{\Lb b} = D_a F_{ab}-D_b F_{L \Lb}
\end{equation}
Apply $L$ on (\ref{LFLbb}), we have
\begin{align*}
 2 L (D_L F_{\Lb b}) &= L(D_a F_{ab})-L(D_b F_{L \Lb})\\
             &= (D_L D_a F_{ab} + D_a F_{(D_L a) b} + D_a F_{a (D_L b)})\\
             &\quad - (D_L D_b F_{L \Lb}+ D_b F_{(D_L L) \Lb} + D_b F_{L (D_L \Lb)})\\
             &= D_L D_a F_{ab}-D_L D_b F_{L \Lb}\\
             &=(D_a D_L F_{ab}-R_{Laa}{}^\rho F_{\rho b}-R_{Lab}{}^\rho F_{a\rho})\\
             &\quad -(D_b D_L F_{L\Lb}-R_{LbL}{}^\rho F_{\rho \Lb}-R_{La\Lb}{}^\rho F_{L\rho})\\
             &=D_a D_L F_{ab}-D_b D_L F_{L\Lb}\\
             &=[e_a(D_L F_{ab})-D_L F_{(D_a e_a)b}-D_L F_{a(D_a e_b)}]\\
             &\quad - [e_b(D_L F_{L\Lb})-D_L F_{(D_b L)\Lb}-D_L F_{L(D_a \Lb)}]\\
             &=0.
\end{align*}
So we have
\begin{equation}\label{LLFLbb}
L(D_L F_{\Lb b} )= 0
\end{equation}
Now we are ready to show $\LL_K F_{\Lb b}=0.$
\begin{align*}
 \LL_K F_{\Lb b} &= D_K F_{\Lb b} + g^{\rho\de}D_\Lb K_\rho F_{\de b}+ g^{\rho\de}D_b K_\rho F_{\Lb \de}\\
        &= \ub D_L F_{\Lb b} + D_\Lb K_a F_{a b} - D_\Lb K_L F_{\Lb b}-D_b K_\Lb F_{\Lb L}\\
        &= \ub D_L F_{\Lb b} -\ub \ze_a F_{a b} - F_{\Lb b}-\ub \ze_b F_{\Lb L}
\end{align*}
In particular, this shows $\LL_K F_{\Lb b} = 0$ on $\S$.
Notice that $L(F_{Lb})=L(F_{ab})=L(F_{L\Lb})=0$, now apply $L$ on $\LL_K F_{\Lb b}$, so we have
\begin{align*}
 L(\LL_K F_{\Lb b}) &= L(\ub D_L F_{\Lb b}) -L(\ub \ze_a F_{a b}) - L(F_{\Lb b})-L(\ub \ze_b F_{\Lb L})\\
            &\stackrel{(\ref{LLFLbb})}{=} D_L F_{\Lb b} - \ze_a F_{a b}-[D_L F_{\Lb b} + F_{(D_L \Lb)b}+F_{\Lb(D_L b)}]-\ze_b F_{\Lb L}\\
            &=0
\end{align*}
Now solving this ordinary differential equation on $\N^+$ completes the proof.
\end{proof}
\begin{remark}
It follows from the previous computation that $D_\a \pi_{\b\ga} =
0$ on $\N^+$. That's the nature of hyperbolic equations with
initial on a characteristic surface. In fact, $D_a \pi_{\a\b} = 0$
and $D_L \pi_{\a\b} = 0$ trivially comes from the fact that
$\pi_{\a\b} = 0$ on $\N^+$; to see $D_3 \pi_{\a\b} = 0$, we need
to investigate the first equation in (\ref{pi}), utilizing
$\pi_{\a\b} = 0$ and $\LL_K F=0$, it gives
$$D_4 D_3 \pi_{\a\b}+D_3 D_4 \pi_{\a\b}=0.$$
Combined the curvature identity
$$D_4 D_3 \pi_{\a\b}-D_3 D_4 \pi_{\a\b}= -R_{34\a}{}^\rho \pi_{\rho\b}--R_{34\b}{}^\rho \pi_{\a\rho}=0,$$
it gives $L(D_3 \pi_{\a\b})$=0. So $D_3 \pi_{\a\b}=0$ follows from the fact that it vanishes on $\S$.
\end{remark}

The last statement of Proposition \ref{interior}, $[\Lb, K]=-\Lb$
in the domain of dependence, follows from the fact that
\begin{equation}\label{LK}
\left\{ \begin{array}{rl}
    D_\Lb W &= -D_W \Lb \quad \text{where} \quad W = [\Lb, K] + \Lb, \\
        W&=0 \qquad \qquad on \qquad \N^+ \cap \O\\
        \end{array}\right.
\end{equation}
We first prove this ordinary differential equation holds. Since
$K$ is Killing vector field, we know that for arbitrary vector
fields $X$ and $Y$, we have
\begin{equation*}
 \LL_K (D_X Y) = D_X (\LL_K Y) + D_{\LL_K X} Y.
\end{equation*}
Therefore,
\begin{align*}
 D_\Lb W &= D_\Lb (-\LL_K \Lb + \Lb)=-D_\Lb (\LL_K \Lb)=-(\LL_K(D_\Lb \Lb) - D_{\LL_K \Lb} \Lb)\\
    &= D_{\LL_K \Lb} \Lb = -D_{[\Lb, K] +\Lb} \Lb = -D_W \Lb
\end{align*}
It remains to show $W=0$ on $\N^+$.
\begin{align*}
 W &= D_\Lb K -D_K \Lb + \Lb\\
   &= D_\Lb K -\ub D_L \Lb + \Lb
\end{align*}
Since we have already computed the components $D_3 K_\a$, it's
almost trivial to check $W=0$ on $\N^+$. This completes the proof
of Proposition \ref{interior}.

\section{Hawking vector field outside the black hole}\label{outside}

In the previous section we have constructed the Hawking vector
field $K$ inside the black hole region. To be able to extend it
outside the black hole, because the characteristic initial value
problem is ill-posed in this region, as we explained in the
introduction, we need to rely on a completely different strategy.
The idea is, instead of solving a hyperbolic system, we now can
solve $[\Lb, K]=-\Lb$ for $K$. This ordinary differential equation
is well-posed in the complement of the domain of dependence.
That's how $K$ is constructed. Let $\phi_t$ be the one parameter
diffeomorphisms generated by $K$. When $t$ is small, we show that
$(g,F)$ and $(\phi_t^* g, \phi_t^* F)$ they both verify
Einstein-Maxwell equations and they coincide on $\N^+ \cup \N^-$.
We show that the must be coincide in a full neighborhood of $\S$.
In particular, it shows $K$ is Killing. So it's the Hawking vector
field. In the vacuum case, this is due to Alexakis, Ionescu and
Klainerman, see \cite{AIK}.

To realize this strategy, we first define a vector field $K'$ by
setting $K'=\ub L$ on $\N^+ \cap \O$ and solving the ordinary
differential equation $[\Lb,K']=-\Lb$. The vector field $K'$ is
well-defined and smooth in a small neighborhood of $S$ (since
$\Lb\neq 0$ on $S$)  and coincides with $K$ in $\I^{++} \cup
\I^{--}$  in $\O$. Thus $K:=K'$ defines the desired extension.
This proves the following:
\begin{lemma}\label{extendK}
There exists a smooth extension of the vector field $\K$ to a full
neighborhood $\O$ of $S$ such that
\begin{equation}\label{constructK}
[\Lb,K]=-\Lb \qquad \text{ in } \O.
\end{equation}
\end{lemma}

Let $g_t = \phi_t^* g$ and $\Lb_t = (\phi_{-t})_* \Lb$. In view of
the definition (\ref{constructK}) of $K$, we know
\begin{equation*}
\frac{d}{d t} \Lb_t = - \Lb_t.
\end{equation*}
It implies that
\begin{equation}\label{phiLb}
 \Lb_t = e^{-t} \Lb.
\end{equation}

Let $D^t$ be the Levi-Civita connection of $g_t$, by the tensorial
nature, we know that $D^t_{\Lb_t} \Lb_t = D_\Lb \Lb =0$,
(\ref{phiLb}) infers that $0=D^t_{\Lb_t} \Lb_t = e^{-2t} D^t_\Lb
\Lb$. This proves the following

\begin{lemma}\label{DtLbLb}
Assume $K$ is a smooth vector field constructed in
\eqref{constructK} and $D^t$ the covariant derivative induced by
the metric $\phi_t^* g$. Then,
\begin{equation}\label{dlblb}
 D^t_\Lb \Lb = 0 \qquad \text{in a full neighborhood of } \S.
\end{equation}
\end{lemma}
To summarize, let $F_t = \phi_t^* F$, we have a family of metrics
and 2 forms $(g_t, F_t)$ which verify the Einstein-Maxwell
equations \eqref{E-M} in the domain of dependence of $\N^+ \cup
\N^-$ and such that $D^t_\Lb \Lb = 0$. So the Theorem \ref{first}
is an immediate consequence of the following uniqueness statement:

\begin{proposition}\label{unique}
Assume in a full neighborhood $\O$ of $\S$, $g'$ is a smooth
Lorentzian metric and $F'$ is a smooth 2 form,  such that
$(g',F')$ solves Einstein-Maxwell equations \eqref{E-M}. If
\begin{equation*}
g'=g \quad \text{ in } (\I^{++}\cup \I^{--}) \cap \O \quad \text{and} \quad D'_{\Lb}\Lb=0 \text{ in } \O,
\end{equation*}
where $D'$ denotes the Levi-Civita connection of the metric $g'$.
Then $g'=g$ and $F'=F$ in a full neighborhood $\O' \subset \O$ of
$\S$.
\end{proposition}

The similar proposition for Einstein vacuum space-times was first
proved in \cite{Al}. A simplified version can be found in
\cite{AIK}. In \cite{IK}, the authors proved uniqueness results
for covariant semi-linear wave equations of a fixed metric. But
for the uniqueness at the level of metrics, since the
corresponding partial differential equations are quasi-linear, one
has to couple the system with a system of ordinary differential
equations to recover the semi-linearity. In this section, we use
this idea to prove uniqueness for the full curvature tensor and
the electromagnetic field. Since the metric is uniquely determined
by the curvature, that will prove Proposition \ref{unique}.

\begin{proof}
We first derive a system of covariant wave equations for the full
curvature tensor $R_{\a\b\ga\de}$ of the metric $g$ and
$F_{\a\b}$. Recall the second Bianchi identities and once
contracted Bianchi identities:
\begin{equation}\label{2bianchi}
D_\a R_{\b\ga\rho\de} +D_\b R_{\ga\a\rho\de}+D_\ga R_{\a\b\rho\de}=0
\end{equation}

\begin{equation}\label{1con}
D^\a R_{\a\de\b\ga}=D_\ga R_{\b\de}-D_\b R_{\ga\de}
\end{equation}

We apply $D^\a$ on \eqref{2bianchi} and commute derivatives, we have
\begin{align*}
 D^\a D_\a R_{\b\ga\rho\de} &= -[D^\a, D_\b] R_{\ga\a\rho\de}-[D^\a, D_\ga] R_{\a\b\rho\de}- D_\b D^\a R_{\ga\a\rho\de} - D_\ga D^\a R_{\a\b\rho\de}\\
&=R^\a{}_{\b\ga\mu}R^\mu{}_{\a\rho\de} + R^\a{}_{\b\a\mu}R_{\ga}{}^\mu{}_{\rho\de} + R^\a{}_{\b\rho\mu}R_{\ga\a}{}^\mu{}_{\de} + R^\a{}_{\b\de\mu}R_{\ga\a\rho}{}^\mu\\
&\quad +  R^\a{}_{\ga\a\mu}R^\mu{}_{\b\rho\de} + R^\a{}_{\ga\b\mu}R_{\a}{}^\mu{}_{\rho\de} + R^\a{}_{\ga\rho\mu}R_{\a\b}{}^\mu{}_{\de} +R^\a{}_{\ga\de\mu}R_{\a\b\rho}{}^\mu\\
&\quad + D_\b D_\de R_{\rho \ga} +  D_\ga D_\rho R_{\de \b} - D_\b D_\rho R_{\de \ga} - D_\ga D_\de R_{\rho \b}
\end{align*}

To simplify the formulae, without losing information, we will $*$
notation. The expression $A*B$ is a linear combination of tensors,
each formed by starting with $A \otimes B$, using the metric to
take any number of contractions. So the algorithm to get $A*B$ is
independent of the choices of tensors $A$ and $B$ of respective
types.

Schematically, we write it as
\begin{equation}\label{boxR}
 \square_g R_{\a\b\ga\de} = (R*R)_{\a\b\ga\de}+ D_\ga D_\de R_{\a\b}
\end{equation}
We need to compute the Hessian of Ricci tensor. By the
gravitational part of \eqref{E-M}, we have the following
schematically expression:
\begin{align*}
D_\ga D_\de R_{\a\b} &=F_{\b \mu} D_\ga D_\de F_\a{}^\mu  + F_{\a}{}^{\mu} D_\ga D_\de F_{\b \mu} + D_\de F_\a{}^\mu D_\ga F_{\b \mu} + D_\ga F_\a{}^\mu D_\de F_{\b \mu}\\
&\quad-\frac{1}{2} g_{\a\b}(F^{\mu \nu}D_\ga D_\de F_{\mu \nu} + D_\ga F_{\mu \nu} D_\de F^{\mu \nu})\\
&=(F*D^2 F)_{\a \b \ga \de}+ (DF*DF)_{\a \b \ga \de}
\end{align*}

Plug this in \eqref{boxR}, we have
\begin{equation}\label{eqR}
 \square_g R_{\a\b\ga\de} = (R*R)_{\a\b\ga\de}+ (F*D^2 F)_{\a \b \ga \de}+ (DF*DF)_{\a \b \ga \de}
\end{equation}

Apparently, this equation involves two derivatives of $F$. In
principle, the electromagnetic part of the Einstein-Maxwell
equations \eqref{E-M} controls only one derivative of $F$ through
the second order system:
\begin{align*}
 D^\a D_\a F_{\b\ga} &= -D^\a D_\b F_{\ga\a} - D^\a D_\ga F_{\a\b}\\
             &= -[D^\a, D_\b ] F_{\ga\a} - [D^\a,D_\ga] F_{\a \b}-D_\b D^\a F_{\ga\a} -D_\ga D^\a F_{\a\b}\\
             &= R^{\a}{}_{\b\ga\mu} F^{\mu}{}_{\a}+ R^{\a}{}_{\b\a\mu} F_{\ga}{}^{\mu}
\end{align*}
Schematically, it's expressed as
\begin{equation}\label{eqF}
\square_g F_{\a\b} = (R * F)_{\a\b}
\end{equation}
Since for the Einstein-Maxwell equations, the electromagnetic part
of is almost decoupled from the gravitational part, we can
actually control second derivative of $F$ by a cost of one
derivative on the curvature tensor $R_{\a\b\ga\de}$. Let's apply
$D_\rho$ on the second equation of \eqref{E-M} and commute
derivatives:
\begin{align*}
 d (D_\de F)_{\a\b} &= D_\a D_\de F_{\b\ga} + D_\b D_\de F_{\ga\a} + D_\ga D_\de F_{\a\b}\\
            &= [D_\a, D_\de] F_{\b\ga} + [D_\b, D_\de] F_{\ga\a} + [D_\ga, D_\de] F_{\a\b} + D_\rho (D_{[\a}F_{\b\ga]})\\
            &= -R_{\a\de\b\mu}F^{\mu}{}_{\ga}-R_{\a\de\ga\mu}F_\b{}^{\mu}-R_{\b\de\ga\mu}F^{\mu}{}_{\a}-R_{\b\de\a\mu}F_\ga{}^{\mu}-R_{\ga\de\a\mu}F^{\mu}{}_{\b}-R_{\ga\de\b\mu}F_\a{}^{\mu}
\end{align*}
where $d$ stands for the exterior derivative on $2$ forms.
Schematically, it gives
\begin{equation*}
 D_{[\a}(DF)_{\b\ga]} = (R*F)_{\a\b\ga}
\end{equation*}
Similarly, we have
\begin{equation*}
 D^\a(DF)_{\a\b} = (R*F)_\b
\end{equation*}
Apply covariant derivative on these last two equations, it implies
\begin{equation}\label{eqDF}
 \square_g (DF)_{\a\b} = (R * DF)_{\a\b} + (DR *F)_{\a\b}
\end{equation}
We summarize \eqref{eqR}, \eqref{eqF} and \eqref{eqDF} in following system of equations
\begin{equation}\label{eqmain}
\left\{ \begin{array}{rl}
    \square_g R_{\a\b\ga\de} &= (R*R)_{\a\b\ga\de}+ (F*D^2 F)_{\a \b \ga \de}+ (DF*DF)_{\a \b \ga \de}\\
        \square_g F_{\a\b} &= (R * F)_{\a\b}\\
    \square_g (DF)_{\a\b} &= (R * DF)_{\a\b} + (DR *F)_{\a\b}
        \end{array}\right.
\end{equation}
We have a similar covariant system of equations for $R'_{\a\b\ga\de}$ and $F'_{\a\b}$.\\

We'll prove Proposition \ref{unique} in a neighborhood  $\O(p)$ of
a point $p \in \S$ where introduce a fixed coordinate system $x_k$
for $k=1,2,3,4$ such that it is fixed for both metrics $g$ and
$g'$. In the proof we shall keep shrinking the neighborhoods of
$p$; to simplify notations we keep denoting such neighborhoods by
$\O(p)$.

We now fix null frame $\{e_1, e_2, e_3=\Lb, e_4=L\}$ on the null
hypersurface $\N^+ \cap \O(p)$. Since $g$ and $g'$ agree to
infinity order on $\N^+$, this null frame is the same for both
metrics. Recall also that the vector field $\Lb$ is also the same
for both metrics. We use two different Levi-Civita connections to
parallel transport the given null frame along $\Lb$:

\begin{equation}\label{defineframe}
\left\{ \begin{array}{rl}
    D_\Lb v_\a = 0 \qquad \text{with} \quad v_\a = e_\a \quad \text{on} \quad \N^+ \cap \O(p)\\
        D'_\Lb v'_\a = 0 \qquad \text{with} \quad v'_\a = e_\a \quad \text{on} \quad \N^+ \cap \O(p)
        \end{array}\right.
\end{equation}

The frames $\{v_\a\}$ and $\{v'_\a\}$ are smoothly defined in
$\O(p)$. We will express all the geometric quantities in these
frames. Let $g_{\a\b} = g(v_\a, v_\b)$, $g'_{\a\b} = g'(v'_\a,
v'_\b)$. Since $D_\Lb v_\a = D'_\Lb v'_\a =0$,  we know
$\Lb(g_{\a\b})=\Lb(g'_{\a\b})=0$, so $g_{\a\b} = g'_{\a\b}$. It
follows that
\begin{equation}\label{eqG}
h_{\a\b}\stackrel{\Delta}{=}g_{\a\b} = g'_{\a\b} \qquad \Lb(h_{\a\b})=0 \quad \text{in  } \O(p).
\end{equation}
Now define the Christoffel symbols, curvature tensors and their
differences,
\begin{equation*}
\Ga^{\ga}_{\a\b} \stackrel{\Delta}{=} g(D_{v_\a}v_\b, v_\ga),\quad
\Ga'^{\ga}_{\a\b} \stackrel{\Delta}{=} g'(D'_{v'_\a}v'_\b,
v'_\ga), \quad \de \Ga^{\ga}_{\a\b}
\stackrel{\Delta}{=}\Ga'^{\ga}_{\a\b}-\Ga^{\ga}_{\a\b}
\end{equation*}
\begin{equation*}
R_{\a\b\ga\de} \stackrel{\Delta}{=} g(R(v_\a,v_\b)v_\ga, v_\de),
R'_{\a\b\ga\de} \stackrel{\Delta}{=} g'(R'(v'_\a,v'_\b)v'_\ga,
v'_\de), \de R_{\a\b\ga\de} \stackrel{\Delta}{=}
R'_{\a\b\ga\de}-R_{\a\b\ga\de}
\end{equation*}
Clearly, we have $\Ga^{\ga}_{3\b} = \Ga'^{\ga}_{3\b}=\de
\Ga^{\ga}_{3\b}=0$. The fact that $D_\Lb v_\a =0$ allow us to
drive a system of ordinary differential equations for $
\Ga^{\ga}_{\a\b}$ and $\Ga'^{\ga}_{\a\b}$:
\begin{align*}
\Lb(\Ga^{\ga}_{\a\b}) &= \Lb(g(D_{v_\a}v_\b, v_\ga))=g(D_{v_3} D_{v_\a}v_\b, v_\ga)+g(D_{v_\a}v_\b, D_{v_3} v_\ga)\\
              &= R_{3\a\b\ga} + g(D_{[v_3,v_\a]} v_\b, v_\ga)+g(D_{v_\a}v_\b, D_{v_3} v_\ga)\\
              &= R_{3\a\b\ga} + \Ga^{\rho}_{3\a}\Ga^{\ga}_{\rho\b}-\Ga^{\rho}_{\a 3}\Ga^{\ga}_{\rho\b} + g_{\rho \de}\Ga^{\de}_{\a\b} \Ga^{\rho}_{3\ga}
\end{align*}
Schematically, we have
\begin{equation}\label{LbGa}
\Lb(\Ga^{\ga}_{\a\b}) = R_{3\a\b\ga} + (\Ga * \Ga)^\ga _{\a\b}
\end{equation}
\begin{equation}\label{LbGa'}
\Lb(\Ga'^{\ga}_{\a\b}) = R'_{3\a\b\ga} + (\Ga' * \Ga')^\ga _{\a\b}
\end{equation}
We take the difference of \eqref{LbGa} and \eqref{LbGa'}, so we have
\begin{align*}
\Lb(\de \Ga^{\ga}_{\a\b}) &=\de R_{3\a\b\ga} + (\Ga' * \Ga'-\Ga * \Ga)^\ga _{\a\b}\\
              &=\de R_{3\a\b\ga} + (\Ga'* \de \Ga)^\ga _{\a\b} + (\Ga* \de \Ga)^\ga _{\a\b}
\end{align*}
Schematically, we have the following expression:
\begin{equation}\label{LbdGa}
\Lb(\de \Ga) = M_\infty(\de \Ga) + M_\infty(\de R).
\end{equation}

\begin{remark}
In general, given $B=(B_1,..., B_L):\O(p) \to \mathbb{R}^L$ we let
$M_\infty(B):\O(p) \to \mathbb{R}^{L'}$ denote vector-valued
functions of the form $M_\infty(B)_{l'}= \sum_{l=1}^L A_{l'}^l
B_l$, where the coefficients $A_{l'}^l$ are smooth on $\O(p)$. So
\eqref{LbdGa} holds because $g$,$g'$ are fixed smooth metrics.
\end{remark}
Now we also need to express the frames $\{v_\a\}$ and  $\{v'_\a\}$
in terms of the fixed coordinate vector fields $\p_k$ relative to
our local coordinates $x_k$. We define
\begin{equation*}
v_\a = v_\a ^k \p_k, \quad v'_\a = {v'}_\a ^k \p_k, \quad (\de
v)_\a^k = {v'}_\a ^k - {v}_\a ^k
\end{equation*}
Consider $[v_3, v_\a] = -D_{v_\a} v_3 = - \Ga^{\b}_{\a 3}v_\b =-\Ga^{\b}_{\a 3}v^k_\b \p_k$, it implies
\begin{equation*}
v_3^j \p_j (v_\a^k)-v_\a ^j \p_j (v_3^k)=-\Ga^{\b}_{\a 3}v^k_\b
\end{equation*}
i.e.
\begin{equation*}
L(v_\a^k)=\p_j (v_3^k) v_\a ^j -\Ga^{\b}_{\a
3}v^k_\b
\end{equation*}

Now a similar relation holds for ${v'}_\a ^k$, we take the
difference, noticing that $\p_j (v_3^k)$ are fixed functions
(since $v_3 = v'_3 =\Lb$), so
\begin{equation}\label{Lbdv}
 \Lb(\de v) = M_\infty(\de \Ga) + M_\infty(\de v).
\end{equation}
We can also apply coordinate derivatives $\p_k$ on \eqref{LbdGa}
and \eqref{Lbdv}, so we have
\begin{equation}\label{LbddGa}
\Lb(\p \de \Ga) = M_\infty(\de \Ga) + M_\infty(\p \de \Ga) + M_\infty(\de R)+ M_\infty(\p \de R).
\end{equation}
\begin{equation}\label{Lbddv}
 \Lb(\p \de v) = M_\infty(\de \Ga) + M_\infty(\p \de \Ga)+ M_\infty(\de v)+M_\infty(\p \de v).
\end{equation}

Finally, we derive a set covariant of wave equations for $\de R$
and $\de F$, $ \de DF$ which are similarly defined for the
difference of the corresponding quantities. In view of
\eqref{eqmain}, the most difficult terms come from the following
differences
\begin{equation*}
(\square_g - \square_{g'})R, \quad (\square_g - \square_{g'})F
\quad \text{and} \quad (\square_g - \square_{g'})DF
\end{equation*}
For the first one, since $g_{\a\b} = g'_{\a\b}$, it's easy to see
it has the following form
\begin{equation*}
(\square_g - \square_{g'})R =  M_\infty(\de \Ga) + M_\infty(\p \de
\Ga).
\end{equation*}
Similar relations hold for the other terms. Together with
\eqref{LbdGa}, \eqref{Lbdv}, \eqref{LbddGa} and \eqref{Lbddv}, we
have the following system of ordinary-partial differential
equations:

\begin{equation}\label{mainequations}
\left\{ \begin{array}{rl}
    \Lb(\de \Ga) &= M_\infty(\de \Ga) + M_\infty(\de R)\\
        \Lb(\p \de \Ga) &= M_\infty(\de \Ga) + M_\infty(\p \de \Ga) + M_\infty(\de R)+ M_\infty(\p \de R)\\
     \Lb(\de v) &= M_\infty(\de \Ga) + M_\infty(\de v)\\
    \Lb(\p \de v) &= M_\infty(\de \Ga) + M_\infty(\p \de \Ga)+ M_\infty(\de v)+M_\infty(\p \de v)\\
    \square_g \de R  &= M_\infty(\de R) + M_\infty(\de F)+ M_\infty(\de DF)+ M_\infty(\p \de D F)\\
             &\quad +M_\infty(\de \Ga) + M_\infty(\p \de \Ga)\\
        \square_g \de F &= M_\infty(\de R)+M_\infty(\de F)+M_\infty(\de \Ga) + M_\infty(\p \de \Ga)\\
    \square_g \de DF &= M_\infty(\de R) + M_\infty(\de DF) + M_\infty(\p \de R)+ M_\infty(\de F)\\
             &\quad +M_\infty(\de \Ga) + M_\infty(\p \de \Ga)
        \end{array}\right.
\end{equation}

Since in $\I^{++} \cup \I^{--}$, $g = g'$ and $F = F'$, so we know
that, on the bifurcate horizon $\N^+ \cup \N^-$, $\de$, $\Ga$, $\p
\de \Ga$, $\de v$, $\p \de v$, $\de R$, $\de F$ and $\de DF$
vanish. We need one more ingredient to conclude that the previous
system of equations has only zero as its solution. It's the
following uniqueness theorem, based on the Carleman estimates
developed in \cite{IK}, due to Alexakis \cite{Al}, see also Lemma
4.4 of \cite{AIK}.

\begin{proposition}\label{Carleman}
Assume  $G_i,H_j:\O(p) \to \mathbb{R}$ are smooth functions,
$i=1,...,I$, $j=1,...,J$. Let $G=(G_1,...,G_I)$,
$H=(H_1,...,H_J)$, $\p G=(\p_1 G_1,\p_2 G_1,\p_3 G_1,\p_4 G_1,
...,\p_4 G_I )$ and assume that in $\O(p)$,
\begin{equation*}
\begin{cases}
&\square_\g G=\mathcal{M}_\infty(G)+\mathcal{M}_\infty(\partial G)+\mathcal{M}_\infty(H);\\
&\Lb(H)=\mathcal{M}_\infty(G)+\mathcal{M}_\infty(\partial G)+\mathcal{M}_\infty(H).
\end{cases}
\end{equation*}
Assume that $G=0$ and $H=0$ on $(\N^+ \cup \N^-) \cap \O(p)$.
Then, there exists a neighborhood $\O'(p) \subset \O(p)$ of $x_0$
such that $G=0$ and $H=0$ in $(\I^{+-} \cup \I^{-+})\cap \O'(p)$.
\end{proposition}

Apparently, this proposition finishes the proof of Proposition
\ref{unique}, which implies that the vector field $K$ is Killing
in a full neighborhood of $\S$.
\end{proof}

\begin{remark}
The vector field $K$ is time-like outside the black hole, i.e.
$g(K,K) \leq 0$ in $\I^{+-} \cup \I^{-+}$, which follows directly
from the fact that $\Lb (g(K,K)) \geq 0$.
\end{remark}

\section{Rotational Killing vector field}\label{rotation}

The purpose of this section is to prove Theorem \ref{second}. In
addition to the Hawking vector field $K$ we just constructed, we
assume $(\O,g,F)$ has another Killing vector field $T$ such that
it's tangent to $\N^+ \cup \N^-$, non-vanishing on $\S$ and $\LL_K
F = 0$. We need to find a constant $\la$, such that $Z = T + \la
K$ is a rotational Killing vector fields, i.e. all the orbits are
closed.

One needs to study the action of $T$ on the bifurcate sphere $\S$.
Since $T$ is a smooth vector field tangent the bifurcate horizon
$\N^+ \cup \N^-$, it must be tangent to $\S$. We can conclude that
the existence of such a non-vanishing Killing vector field $T$ on
$\S$ forces the restriction of the metric $g$ on $\S$ to be
rotational symmetric thanks to Lemma \ref{lemma1} in the appendix.
In our case $X = T|_\S$ on $\S$ with induced metric from $g$. It
has a period $t_0$. It has two zeroes and we choose one of them,
denoting it by $p\in \S$. To get a space-time rotational vector
field, we need to study $T$ on the black hold boundary $\N^+ \cup
\N^-$. On $\N^+$, we define $\la(T) = \frac{g(T, \Lb)}{g(K,\Lb)}$
which is essentially the $K$ direction of $T$. We prove the
following lemma

\begin{claim}\label{lambdaT}
On $\N^+$,  $L(T)$ is constant along each null geodesic, i.e. $L(\lambda(T))=0$.
\end{claim}
\begin{proof}
We first show that $[T, L]$ is parallel to $L$, i.e. there is a
function $f: \N^+ \to \mathbb{R}$, such that
\begin{equation*}
[T,L] = fL.
\end{equation*}
Since both vectors are tangent to $\N^+$, so is $[T,L]$. It
suffices to show $g([T,L], e_a)=0$.
\begin{align*}
 g([T,L], e_a) &=g(D_T L, e_a) - g(D_L T, e_a) \stackrel{Killing}{=}g(D_T L, e_a) + g(D_a T, L)\\
           &=g(D_T L, e_a) - g(T, D_a L) =\chi(T, e_a)-\chi(e_a, T)=0
\end{align*}

We then show that $L(f)=0$. Since $D_L L=0$ and $T$ is Killing, we have
\begin{align*}
 0 &= \LL_T(D_L L) = D_{\LL_T L} L+ D_L (\LL_T L)\\
   &= D_{fL} L + D_L(fL) =L(f)L
\end{align*}
It implies that $f$ is determined on $\S$. We can assume $f: \S
\to \mathbb{R}$.
\begin{align*}
 f&=fL(\ub)=[T,L](\ub)=-L(T(\ub))
\end{align*}
So
\begin{equation*}
T(\ub)= -f \ub.
\end{equation*}
Now we compute $L(\la(T))$ by recalling that $\Lb$ is the gradient
of $\ub$ under the metric $g$:
\begin{align*}
 L(\la(T)) &= L(\frac{g(T, \Lb)}{g(K,\Lb)})=-L(\frac{T(\ub)}{\ub})\\
       &= L(f)=0
\end{align*}
\end{proof}

Now we can find the rotational vector field $Z$:
\begin{claim}
Let $\lambda = f(p)$, then $Z =T-\lambda K$ is a rotational vector
field with period $t_0$.
\end{claim}
\begin{proof}

Since $K=0$ on $\S$, $Z|_\S = T|_\S$ has the same period $t_0$. We
denote $\psi_t$ the one parameter isometry group generated by $Z$
on space-time. We are going to prove that $\psi_{t_0} = id$ which
concludes the proof of the claim.

We study the action of $\psi_t$ on the null geodesic $\ga$
starting at $p$ and pointing at the $L$ direction. For each $t$,
since $p$ is a fixed point of $\psi_t$ and $\psi_t$ is an
isometry, we know that $\psi_t(\ga) \subset \ga$ is an
reparametrizition of $\ga$ with a possible stretch. In particular,
it implies $Z|_\ga$ is proportional to $K|_\ga$. In view of the
definition of $\lambda$, we know that $Z|\ga = 0$ since we have
subtracted the corresponding portion of $K$ from $T$. So
$\psi_t|_\ga = id$. In particular, $\psi_{t_0}|_{\ga} =id$.

Now we look at the action of $\psi_{t_0}$ on the full tangent
space of $p$. The previous argument shows $(\psi_{t_0})_{*}L = L$.
Since it fixes the whole space slice $\S$, then $(\psi_{t_0})_*
e_a = e_a$. Now using the fact that $\psi_{t_0}$ is an isometry,
we know $\Lb$ is also fixed. So $(\psi_{t_0})_*$ is the identity
map on the tangent space of $p$, now we can use Lemma \ref{lemma2}
in the appendix to conclude that $\psi_{t_0}$ is identity in a
small neighborhood of $p$. Now on can use the compactness of $\S$
and the standard open-closed argument on $\S$ to conclude
$\psi_{t_0}$ is identity map in a small neighborhood of $\S$.
\end{proof}

We need one more claim to finish the proof of Theorem \ref{second}:

\begin{claim}
 $Z$ is the vector field we constructed, then $[Z, K]=0$.
\end{claim}
\begin{proof}
It suffices to show $[T,K]=0$. Since both $K$ and $T$ are Killing,
in view of the fact that all the Killing vector fields on a
manifold form a Lie algebra under $[-,-]$, we know that $W =[T,K]$
is Killing, so it solves the following equation:
\begin{equation}\label{equationW}
 \square_g W_\a = -W_\a{}^\b W_\b
\end{equation}
Once again, due to the well-posedness of the characteristic
initial-value problem, $W=0$ in the domain of dependence follows
from the fact that
\begin{equation*}
W=0 \qquad \text{on} \quad \N^+ \cup \N^-.
\end{equation*}
It is immediate from the calculations in the proof of Claim
\ref{lambdaT}:
\begin{align*}
 W&=[T,K]=[T,\ub L]=\ub[T,L]+T(\ub)L\\
  &=\ub f L - \ub f L =0
\end{align*}
For ill-posed region $\I^{+-} \cup \I^{-+}$, once again the
vanishing of $W$ follows easily from setting $H=0$ in Proposition
\ref{Carleman}.
\end{proof}

\appendix

\section{Two lemmas on geometry}

\begin{lemma}\label{lemma1}

Assume $h$ is a Riemannian metric on the topological sphere $\S^2$
which admits a non-trivial Killing vector field $X$, then
$(\S^2,h)$ is a Riemannian wrapped product $([0,1],dr^2)
\times_{\phi(r)} (\S^1, d\sigma^2)$. In particular, each orbit of
$X$ is closed and has a common period $t_0$.
\end{lemma}

\begin{proof}

First, we observe that, if $X$ is non-trivial, then the set
$Z(X)$, which consists all zeroes of $X$, is discrete. It follows
from the fact that, the zero locus of a Killing vector field is a
disjoint union of totally geodesic sub-manifolds each of even
dimension. Since we are on a surface, the zeroes must be discrete.
In particular, since the $\S^2$ is compact, $X$ has only finite
many zeroes.

The second observations is that, for each zero $p$ of $X$,
$ind_X(p)$ the index of $X$ at $p$ is either $1$ or $-1$. It
following from the fact that, $X$ induces an isometry on $T_p
\S^2$, which is a 2-dimensional rotation. So its index must be $1$
or $-1$.

Now we can apply the Poincar\'e-Hopf index Theorem:
\begin{equation*}
\sum_{p\in Z(X)} ind_X(p) = \chi(\S^2) = 2.
\end{equation*}
The previous observation imply that the cardinal number $|Z(X)|
\geq 2$. We can pick up two points $p,q \in Z(X)$. Now let us fix
a minimal geodesic $\ga(t)$ between $p$ and $q$. Let $\phi_t$ be
the flow generated by $X$. Since on $T_p M$, $(\phi_t)_*$ is a
rotation, it has a period $t_0$. Let $x\neq p, q$ be a point on
$\ga$. We show that the orbit of $x$ under $\phi_t$ is a closed
non-degenerate circle, more precisely, it is exactly the image
$\{\phi_t (x) | t\in[0, t_0)\}$. It trivially holds when $x$ is
close to either $p$ or $q$, i.e. in the normal coordinate of $p$
or $q$, since it will stay on the geodesic sphere which is a
circle around either $p$ or $q$. Since $\ga$ is minimal and
$X(q)=0$, so $\phi_t (\ga)$ is also a minimal geodesic between $p$
and $q$. When $t$ varies, $\phi_t (\ga)$ sweeps the whole $\S^2$,
we know that all points except $q$ is in the normal coordinate of
$p$, so the orbit $x$ is closed. Apparently, this finishes the
proof of the lemma.
\end{proof}

\begin{lemma}\label{lemma2}
Assume $(M, g)$ is a Lorentzian manifold, $\phi: M \to M$ is an
isometry and $p\in M$ is one fixed point of $\phi$. If $\phi_{*p}
=id$, the $\phi = id$ locally around $p$.
\end{lemma}

\begin{proof}
In Riemannian geometry, it's easy since we have the concept of
length; in our case, the difficulty comes from the fact that on
the light-cone, we don't have the concept of length. But the
proposition holds inside light-cone since we can consider the
maximal time-like geodesics. Since locally light-cone is the
boundary of the future of the point $p$, the identity map can be
continued to the boundary.
\end{proof}

 \end{document}